\documentclass[AMA,STIX1COL]{WileyNJD-v2}

\usepackage{subfigure}
\articletype{Article Type}%
\algdef{SE}[DOWHILE]{Do}{doWhile}{\algorithmicdo}[1]{\algorithmicwhile\ #1}%
\allowdisplaybreaks

\received{26 April 2016}
\revised{6 June 2016}
\accepted{6 June 2016}

\begin{document}

\title{A Physics-Informed Learning Framework to Solve the Infinite-Horizon Optimal Control Problem\protect\thanks{This work was supported in part, by NSF under grant Nos. CAREER CPS-$1851588$,   CPS-$2227185$, S\&AS-$1849198$, and SLES-$2415479$, by ARO under grant No. W$911$NF-$24-1-0174$, and by the Onassis Foundation-Scholarship ID: F ZQ $064-1/2020-2021$.}}

\author[1]{Filippos Fotiadis*}

\author[2]{Kyriakos G. Vamvoudakis}

\authormark{F. Fotiadis, K. G. Vamvoudakis}

\address[1]{\orgdiv{The Oden Institute for Computational Engineering \& Sciences}, \orgname{University of Texas at Austin}, 
\orgaddress{\state{Texas}, \country{USA}}}

\address[2]{\orgdiv{Daniel Guggenheim School of Aerospace Engineering}, \orgname{Georgia Tech}, \orgaddress{\state{Georgia}, \country{USA}}}

\corres{*Filippos Fotiadis \email{ffotiadis@utexas.edu}}

\presentaddress{
201 E 24th St, Austin, TX 78712, USA }

\abstract[Summary]{
	We propose a physics-informed neural networks (PINNs) framework to solve the infinite-horizon optimal control problem of nonlinear systems. In particular, since PINNs are generally able to solve a class of partial differential equations (PDEs), they can be employed to learn the value function of the infinite-horizon optimal control problem via solving the associated steady-state Hamilton-Jacobi-Bellman (HJB) equation.
    However, an issue here is that the steady-state HJB equation generally yields multiple solutions; hence if PINNs are directly employed to it, they may end up approximating a solution that is different from the optimal value function of the problem. We tackle this by instead applying PINNs to a finite-horizon variant of the steady-state HJB that has a unique solution, and which uniformly approximates the optimal value function as the horizon increases. An algorithm to verify if the chosen horizon is large enough is also given, as well as a method to extend it -- with reduced computations and robustness to approximation errors -- in case it is not.  Unlike many existing methods, the proposed technique works well with non-polynomial basis functions, does not require prior knowledge of a stabilizing controller, and does not perform iterative policy evaluations.  Simulations are performed, which verify and clarify theoretical findings.
}


\jnlcitation{\cname{%
\author{F. Fotiadis}, and
\author{K. G. Vamvoudakis}} (\cyear{2025}), 
\ctitle{A Physics-Informed Learning Framework to Solve the Infinite-Horizon Optimal Control Problem}}

\maketitle

\section{Introduction}

The infinite-horizon optimal control problem \cite{lewis2012optimal}, also known as the optimal stabilization problem \cite{haddad2011nonlinear}, has been a long-standing issue in control theory. Its main purpose is to find the control policy that stabilizes a dynamical system to the origin, while achieving a desired balance between the rate of convergence and the control effort expended in the closed loop. It yields a closed-form solution when the system dynamics are linear and the desired balance expressed by a quadratic function \cite{lewis2012optimal}, however, no such closed-form solution exists when the system dynamics are nonlinear or if the cost function is non-quadratic. 

At the heart of nonlinear/non-quadratic optimal stabilization is the so-called steady-state HJB equation. The importance of this equation lies in that it is solved by the optimal value function of the problem, and this function is crucial in obtaining the optimal control policy for the system. Nevertheless, the steady-state HJB is a nonlinear PDE that is difficult -- and almost always impossible -- to solve analytically. For this reason, a lot of the literature focuses on trying to solve it approximately, typically using actor-critic methods \cite{kiumarsi2017optimal, abu2005nearly, vamvoudakis2010online, modares2013adaptive, kamalapurkar2016efficient, mazouchi2021data}, many of which have also been extended in the context of multi-agent systems and data-driven control \cite{jiang2014robust, ruizhuo, zuo1, zuo2, modares}. 

One of the main methods for approximately solving the steady-state HJB equation for the optimal value function is what is known as successive approximations, or policy iteration (PI) \cite{beard1998approximate, kiumarsi2017optimal}. This procedure begins with a stabilizing control policy for the system, then proceeds to iteratively evaluate and improve it till convergence. However, the requirement of knowing a policy that asymptotically stabilizes the system to the origin can be restrictive, especially for a nonlinear system. Moreover, the unavoidable need to perform an iterative policy evaluation procedure can lead to an increased computational burden. Finally, polynomial basis functions are often required in the PI procedure \cite{beard1998approximate}, which can lead to severe numerical issues when higher-order monomials are used in the basis.

A different, emerging approach for solving a general class of PDEs, is that of PINNs \cite{raissi2019physics}. The main idea is that the function to be solved for in the underlying PDE can be uniformly approximated by a neural network (NN), the weights of which are trained to force the flow and boundary equations of the PDE to hold as closely as possible over some data. Interestingly, this idea has already been used in the literature to approximate optimal value functions through their associated steady-state HJB equations \cite{furfaro2022physics}, however, a critical issue is that the steady-state HJB equation generally yields more than one solution. This can be problematic in practice, as it can cause the PINN to converge to a solution of the HJB that is unrelated to the optimal value function, or to stall very easily during training. Interestingly, the reference    \cite{raissi2019physics} also points out that PINNs generally work well only when applied to PDEs with unique solutions.

\textit{Contributions:} Motivated by this discussion,  we propose a PINNs framework for solving the infinite-horizon optimal control problem, which does not suffer from the issue of possibly converging to unwanted solutions of the steady-state HJB equation. To avoid these solutions, instead of directly applying PINNs to the steady-state HJB, we apply them to its finite-horizon variant that has a unique solution, and rigorously prove that this solution uniformly approximates the optimal value function when the horizon length is large enough. We also provide: a) a method to evaluate whether the horizon length is indeed large enough; b) a method to increase it (in case it is not) with reduced computations; and c) a proof that 
approximation errors do not necessarily accumulate when if the horizon length is increased multiple times. Unlike PI, the proposed technique does not require knowledge of an initial stabilizing control policy, or computationally expensive iterative policy evaluations. 

A preliminary version of this paper appeared in \cite{fotiadis2023physics}. However, \cite{fotiadis2023physics} did not include proofs of the main results, did not study robustness to approximation errors, and was restricted to finite-horizon approximations of the steady-state HJB with terminal costs equal to zero.

\textit{Structure:} The rest of the paper is structured as follows. Section \ref{sec:prob} formulates the problem of infinite-horizon optimal control with PINNs, and a solution is proposed in Section \ref{sec:PINNs} by applying them to the finite-horizon HJB. In Section \ref{sec:unif}, it is proved that the solution of the finite-horizon HJB uniformly approximates the infinite horizon value function when the horizon is large enough. In Section \ref{sec:horizon}, a method to evaluate and extend the horizon length is provided, and the effect of approximation errors is studied. Finally, simulations are conducted in Section \ref{sec:sim}, and the paper is concluded in Section \ref{sec:conc}.

\textit{Notation:} The set $\mathbb{R}$ denotes the set of real numbers. The operator $\nabla_z$ is used to denote the gradient of a function, with respect to the argument implied by $z$.

\section{Problem Formulation}\label{sec:prob}

Consider a continuous-time affine nonlinear system of the form:
\begin{equation}\label{eq:sys}
\dot{x}(t)=f(x(t))+g(x(t))u(t),~x(0)=x_0, ~t\ge0,
\end{equation}
where $x(t)\in\mathbb{R}^n$ denotes the state with initial value $x_0\in\mathbb{R}^n$, $u(t)\in\mathbb{R}^m$ is the control input, and $f:\mathbb{R}^n\rightarrow\mathbb{R}^n$, $g:\mathbb{R}^n\rightarrow\mathbb{R}^{n\times m}$ are the system's drift and input dynamics functions, respectively. We assume that $f,~g$ are locally Lipschitz, so that existence and uniqueness of solutions to \eqref{eq:sys} can be established. Moreover, we assume that $f(0)=0$, so that the origin is an equilibrium point of the uncontrolled version of \eqref{eq:sys}.

The infinite horizon optimal control problem regards finding the feedback control policy that stabilizes system \eqref{eq:sys} to the origin, while having some form of optimality. By optimality, we mean that it minimizes the cost functional
\begin{equation}\label{eq:cost}
J(\mu;~x_0)=\int_0^\infty\left(Q(x)+\mu^\textrm{T}(x)R\mu(x) \right)\textrm{d}\tau,
\end{equation}
with respect to $\mu:\mathbb{R}^n\rightarrow\mathbb{R}^{m}$, where $Q:\mathbb{R}^n\rightarrow\mathbb{R}_+$ is a positive definite function penalizing the state trajectories, $R\in\mathbb{R}^{m\times m}$ is a positive definite matrix penalizing the control inputs, and the integration is over the trajectories of \eqref{eq:sys} under $u=\mu$. The integral \eqref{eq:cost} is well-defined in any compact set $\Omega\subseteq\mathbb{R}^n$ as long as $\mu$ is \textit{admissible} in that set.

\begin{definition}
A control policy $\mu:\mathbb{R}^n\rightarrow\mathbb{R}^m$ is defined to be admissible on a set $\Omega\subseteq\mathbb{R}^n$, and denoted as $\mu\in\Psi(\Omega)$, if it is continuous with $\mu(0)=0$, and if $u=\mu$ asymptotically stabilizes \eqref{eq:cost} to the origin for any $x_0\in\Omega$ while $J(\mu;~x_0)$ is finite.
\end{definition}

Denote the optimal control policy on a compact set $\Omega\subseteq\mathbb{R}^n$ as
\begin{equation*}
\mu^\star(x)=\arg\min_{\mu\in\Psi(\Omega)} J(\mu;~x),
\end{equation*}
and the corresponding optimal value function as $V^\star(\cdot)=J(\mu^\star;~\cdot)$. 
Then, by defining the Hamiltonian
\begin{equation*}
H(x,\mu(x),\nabla_xV(x))=\nabla_xV^\textrm{T}(x)(f(x)+g(x)\mu(x))+Q(x)+\mu^\textrm{T}(x)R\mu(x),
\end{equation*}
and using the stationarity condition $\frac{\partial H}{\partial \mu}=0$,  the optimal control policy can be shown to satisfy \cite{lewis2012optimal}:
\begin{equation}\label{eq:mustar}
\mu^\star(x)=-\frac{1}{2}R^{-1}g^\textrm{T}(x)\nabla_xV^{\star}(x).
\end{equation}
In addition, the optimal value function can be shown to satisfy the following PDE, also known as the steady-state HJB equation \cite{lewis2012optimal}:
\begin{equation}\label{eq:HJB}
\nabla_xV^{\star\textrm{T}}(x)f(x)+Q(x)-\frac{1}{4}\nabla_xV^{\star\textrm{T}}(x)g(x)R^{-1}g^\textrm{T}(x)\nabla_xV^\star(x)=0,~V^\star(0)=0.
\end{equation}
Clearly, if one can solve \eqref{eq:HJB} for the optimal value function $V^\star$, then it is straightforward to obtain the optimal control policy from the formula \eqref{eq:mustar}. Nevertheless, obtaining $V^\star$ from \eqref{eq:HJB} is not an easy task. This is because the PDE \eqref{eq:HJB} is nonlinear, yields multiple solutions even in the linear-quadratic case \cite{bernstein2009matrix},  and is thus very difficult to solve analytically. For this reason, the majority of the literature focuses on solving this equation approximately \cite{vamvoudakis2010online, abu2005nearly, jiang2014robust}. However, most of the existing approximate solution concepts either rely on prior knowledge of an admissible control input for \eqref{eq:sys}, on the use of polynomial basis functions for approximation, or on the execution of computationally expensive iterations, all of which can be restrictive requirements in practice.

Motivated by the aforementioned facts, we seek an alternative method to approximately solve the PDE \eqref{eq:HJB} for $V^\star$. Towards this end, we draw inspiration from the method of the so-called physics-informed neural networks, which have been extensively used in the literature to solve a general class of PDEs \cite{raissi2019physics}. Using these networks, we can exploit the structural aspect of the Hamilton-Jacobi Bellman equation, which is informed about the physics of the system through the system dynamics functions $f, g$, in order to obtain the optimal value function and controller.

\begin{remark}
While we could allow $\Omega$ to be any arbitrary set, within this work we restrict our attention to $\Omega$ being compact. This is because we will employ a neural network to approximate $V^\star$ and $\mu^\star$, and the uniform approximation properties of neural networks are valid only on compact sets. 
\end{remark}

\section{A Finite-Horizon Approximation with PINNs}\label{sec:PINNs}

In this section, we propose a procedure for approximating the infinite-horizon optimal value function $V^\star$ via the use of PINNs.  We note that one of the main difficulties in this direction is the fact that the steady-state HJB \eqref{eq:HJB} is not in the standard form for which PINNs have been developed \cite{raissi2019physics}.  In particular, the time-derivative term in the steady-state HJB equation \eqref{eq:HJB} is missing, and such a term is critical because it ensures uniqueness of solutions. Without such a term, solving \eqref{eq:HJB} is akin to an equilibrium search problem in ordinary differential equations, which, while having unique solutions given certain regularity assumptions, can admit multiple equilibrium points. 
Accordingly, applying PINNs to directly solve \eqref{eq:HJB} could cause the training procedure of PINNs to either stall or converge to a solution completely unrelated to the value function $V^\star$.

To deal with these two issues, note that \eqref{eq:HJB} is the steady-state version of another, time-varying HJB equation that admits a unique solution. In particular, it is the steady-state version of the PDE:
\begin{equation}\label{eq:HJBT}
0=\nabla_t V_T(x,t)+\nabla_x V_T^{\textrm{T}}(x,t)f(x)+Q(x)-\frac{1}{4}\nabla_x V_T^{\textrm{T}}(x,t)g(x)R^{-1}g^\textrm{T}(x)\nabla_x V_T(x,t),\quad V_T(x,T)=\phi(x),
\end{equation}
where $T>0$ is a positive horizon, and $\phi:\mathbb{R}^n\rightarrow\mathbb{R}_+$ is a positive semidefinite function. The unique solution $V_T$ of \eqref{eq:HJBT} is the optimal value function of the \textit{finite-horizon} approximation of \eqref{eq:cost}, namely of:
\begin{equation}\label{eq:costT}
    J_T(\mu;~x_0)=\int_{0}^{T}\Big(Q(x(\tau))+r\left(\mu(x(\tau),\tau)\right)\Big)\textrm{d}\tau+\phi(x(T)),
\end{equation}
where we note that the policy $\mu$ here is allowed to be time-varying. Hence, as $T$ increases, it is expected that the finite-horizon value function $V_T(\cdot,0)$ will pointwise approximate the infinite-horizon value function $V^\star$. In addition, it is expected that the control policy $\mu_T:\mathbb{R}^n\times[0,T]\rightarrow\mathbb{R}^m$ that minimizes \eqref{eq:costT}, and is given by: 
\begin{equation}\label{eq:muTstar}
\mu_T(x,t)=-\frac{1}{2}R^{-1}g^\textrm{T}(x)\nabla_x {V_T}(x,t),
\end{equation}
will approach $\mu^\star$ as $T$ increases, i.e.  that $\mu_T(x,0)$ will converge pointwise to $\mu^\star(x)$ for all $x\in\mathbb{R}^n$. Hence, one can proceed to approximate $\mu^\star$ and $V^\star$ by solving \eqref{eq:HJBT} via the method of PINNs for a sufficiently large horizon $T>0$. In the following sections, we will mathematically show that this approximation is valid, i.e., that $V_T(\cdot,0)$ indeed uniformly approximates $V^\star$ over $\Omega$ as $T$ increases, and a similar result holds for the optimal control.

Towards approximating the solution $V_T$ of the PDE \eqref{eq:HJBT} using PINNs,  we follow \cite{raissi2019physics} and consider the flow residual function on a compact set $\Omega\subseteq \mathbb{R}^n$ for any $v_T:\Omega\times[0,T]\rightarrow\mathbb{R}$:
\begin{equation*}
F_e(x, t; v_T)=\nabla_t v_T(x,t)+\nabla_x v_T^{\textrm{T}}(x,t)f(x)+Q(x)-\frac{1}{4}\nabla_x v_T^{\textrm{T}}(x,t)g(x)R^{-1}g^\textrm{T}(x)\nabla_x v_T(x,t),~x\in\Omega,t\in[0,T].
\end{equation*}
This residual informs us about how close $v_T$ is to satisfying the time-dependent HJB equation \eqref{eq:HJBT}, and is identically zero if and only if $v_T$ is a solution to \eqref{eq:HJBT} for any boundary condition. Similarly, we consider the following boundary residual function on $\Omega$:
\begin{equation*}
F_b(x; v_T)=v_T(x,T)-\phi(x),~x\in\Omega,
\end{equation*}
which is equal to zero if and only if $v_T$ satisfies the boundary condition of \eqref{eq:HJBT}. Hence, it follows that $v_T$ is a solution to \eqref{eq:HJBT} if and only if $F_e(\cdot,\cdot;v_T)$ and $F_b(\cdot;v_T)$ are identically zero. Finally, one can define the following problem-specific residual:
\begin{equation*}
F_{in}(t; v_T)=v_T(0,t),~t\in[0,T],
\end{equation*}
which, for any value function associated with the optimal stabilization problem, should be equal to zero for all $t\in[0,T]$.

To find a solution to \eqref{eq:HJBT}, a neural network $\hat{v}_T(\cdot,\cdot; w):\Omega\times[0,T]\rightarrow \mathbb{R}$ is constructed, where $w\in\mathbb{R}^N$ denote the network's parameters, and $\Omega\subseteq \mathbb{R}^n$ is a compact set where approximation will take place.  Such a neural network can have either one or multiple layers, and activation functions that could be sigmoids, RELUs, or radial-basis functions. 
The parameters $w$ are trained by attempting to force the residuals $F_e(x, t; \hat{v}_T)$, $F_b(x; \hat{v}_T)$, $F_{in}(t; \hat{v}_T)$ to be zero across a data set $(x^i_e,t^i_e)\in\Omega\times[0,T]$, $i=1,\ldots,N_e$,  $x^i_b\in\Omega$, $i=1,\ldots,N_b$, and $t^i_{in}\in[0,T]$, $i=1,\ldots,N_{in}$. This is done by defining the following mean square error (MSE)
\begin{equation}\label{eq:MSE}
\textrm{MSE}=\textrm{MSE}_e+\textrm{MSE}_b+\textrm{MSE}_{in},
\end{equation}
where 
\begin{align*}
\textrm{MSE}_e&=\frac{1}{N_e}\sum_{i=1}^{N_e}F_e(x^i_e, t^i_e; \hat{v}_T)^2,\\
\textrm{MSE}_b&=\frac{1}{N_b}\sum_{i=1}^{N_b}F_b(x^i_b; \hat{v}_T)^2,\\
\textrm{MSE}_{in}&=\frac{1}{N_{in}}\sum_{i=1}^{N_{in}}F_{in}(t^i_{in}; \hat{v}_T)^2
\end{align*}
and training $w$ so that the MSE given by \eqref{eq:MSE} is minimized.  This training can take place using off-the-shelf optimization algorithms, such as Adam \cite{kingma2014adam, raissi2019physics}.
 The name \textit{physics informed neural network} is usually used for the function $F_e(\cdot,\cdot,\hat{v}_T)$, because it is essentially a neural network with the same parameters $w$ as $\hat{v}_T$, but whose activation functions have been ``informed'' about the physics of the system through the underlying HJB equation.  Although strict theoretical guarantees of convergence do not exist for the method of PINNs, they have been empirically shown to perform well when the solution to the underlying PDE is unique and the neural network architecture is expressive enough \cite{raissi2019physics}. We describe the overall physics-informed learning procedure in Algorithm \ref{al:learning}. 

\begin{algorithm}[!t]
\caption{Physics-Informed Learning for Optimal Control}
\begin{algorithmic}[1]
\Procedure{}{}
\State Choose sufficiently large horizon length $T>0$, area of approximation $\Omega\subseteq\mathbb{R}^n$, and number of grid points $N_e, N_b, N_{in}>0$.
\State Sample points $(x^i_e,t^i_e)\in\Omega\times[0,T]$, $i=1,\ldots,N_e$,  $x^i_b\in\Omega$, $i=1,\ldots,N_b$, and $t^i_{in}\in[0,T]$, $i=1,\ldots,N_{in}$.
\State Choose the number of layers and basis functions for the neural network and construct $\hat{v}(\cdot, \cdot;~w)$.
\State Train $w$ to minimize MSE in \eqref{eq:MSE} using Adam\cite{kingma2014adam}.
\State Compute approximate optimal controller as $\hat{\mu}_T(x,t)=-\frac{1}{2}R^{-1}g^\textrm{T}(x)\nabla_x {\hat{v}_T}(x,t; w).$
\EndProcedure
\end{algorithmic}\label{al:learning}
\end{algorithm}

Apparently, in the context of PINNs, dealing with the HJB equation \eqref{eq:HJBT} is much more convenient than handling the HJB \eqref{eq:HJB}. Unlike \eqref{eq:HJB}, the HJB equation \eqref{eq:HJBT} is exactly the form for which PINNs have been developed \cite{raissi2019physics}. In addition, given that $V_T$ is continuously differentiable, then it is certain that \eqref{eq:HJBT} admits only one solution, equal to $V_T$; hence, the possibility of the PINN approximating a function entirely different than $V_T$ is excluded.  Nevertheless, simply substituting the infinite-horizon HJB with a finite-horizon one with a large horizon, and anticipating that the solution $V_T(\cdot,0)$ of the latter will be close to the solution $V^\star$ of the former, could possibly prove to be a careless action. In that respect, a couple of important questions must be answered to validate such an approach:

\begin{enumerate}
\item As the horizon $T$ increases, do the functions $V_T(\cdot,0)$ and $\mu_T(\cdot,0)$ provide \textit{uniform} approximations of the optimal value function $V^\star(\cdot)$ and control $\mu^\star(\cdot)$?


\item How can one evaluate, for a fixed $T$, whether the derived control policy $\mu_T(\cdot,0)$ is close enough to the infinite-horizon optimal control policy $\mu^\star(\cdot)$? In addition, if it is concluded that $\mu_T(\cdot,0)$ and $V_T(\cdot,0)$ are not close enough to $\mu^\star(\cdot)$ and $V^\star(\cdot)$, how can these functions be used to obtain better ones without completely restarting the training process of the PINN?
\end{enumerate}

The remainder of this paper is focused on answering these significant questions.

\begin{remark}
    While numerical methods could be used to approximately solve \eqref{eq:HJBT}, those have much higher computational complexity as they rely on very fine partitioning of the state space in order to avoid stability issues \cite{raissi2019physics}. At the same time, they cannot generalize outside the compact set $\Omega$ the same way that neural networks do.
\end{remark}

\section{Uniform Approximation of $V^\star$ and $\mu^\star$ using the Finite-Horizon HJB}\label{sec:unif}

The purpose of this section is to investigate whether $V_T(\cdot,0)$ and $\mu_T(\cdot,0)$ provide uniform approximations of the optimal value function $V^\star(\cdot)$ and control $\mu^\star(\cdot)$ as $T$ increases. Towards this end, notice that for any positive semidefinite function $\psi:\mathbb{R}^n\rightarrow\mathbb{R}$, the following inequality is true:
\begin{equation}\label{eq:sandwich}
V_T(\cdot,0; 0)\le V_T(\cdot,0; \psi) \le V_T(\cdot,0; \psi+V^\star),
\end{equation}
where, by abusing notation, the argument after the semicolon denotes the value of the terminal cost in \eqref{eq:costT}. That is, $V_T(\cdot,0; 0)$ is the optimal value of the cost \eqref{eq:costT} with $\phi\equiv0$, $V_T(\cdot,0; \psi)$ is the optimal value for $\phi\equiv\psi$, and $V_T(\cdot,0; \psi+V^\star)$ is the optimal value for $\phi\equiv\psi+V^\star$. Hence, if one can prove that $V_T(\cdot,0; 0)$ and $V_T(\cdot,0; \psi+V^\star)$ converge uniformly to $V^\star$, the same will also hold for $V_T(\cdot,0; \psi)$, i.e., for $V_T(\cdot,0)$ for any arbitrary terminal cost. The rest of this section is focused on proving these facts.

Before proceeding, and whenever necessary to increase clarity, we will denote as $x^\star$ the trajectories of \eqref{eq:sys} generated by the control $\mu^\star$, and as $x_T$ the trajectories of \eqref{eq:sys} generated by the control $\mu_T$. The notation $x_0$ will still be reserved for the initial condition.

\subsection{Uniform Approximation with $\phi\equiv0$}

In this subsection, we prove that for $\phi\equiv0$, $V_T(\cdot,0)$ converges uniformly to $V^\star$ on $\Omega$. Towards this end, we will first need a few intermediate results about the behavior of $V_T$ (for $\phi\equiv0)$ with respect to $T$.

\begin{lemma}\label{le:inc}
Assume that $\phi\equiv0$. Then, for all $x\in\mathbb{R}^n$, the sequence $V_{T}(x,0)$ is increasing with respect to $T$ and upper bounded by $V^\star(x)$, i.e., for every real $T_2\ge T_1 >0$, it holds that:
\begin{equation*}
V_{T_1}(x,0)\le V_{T_2}(x,0)\le V^\star(x),~\forall x\in\mathbb{R}^n.
\end{equation*}
In addition, $\lim_{T\rightarrow\infty}V_T(x,0)=L(x)$ for some Lebesgue measurable function $L:\mathbb{R}^n\rightarrow\mathbb{R}$.
\end{lemma}
\begin{proof}
Let $0<T_1<T_2<\infty$. For every $x_0\in \mathbb{R}^n$, we have:
\begin{align}
\nonumber \hspace{-2mm} V_{T_2}(x_0,0)=J_{T_2}(\mu_{T_2};~x_0)&=\int_{0}^{T_2}\Big(Q(x_{T_2}(\tau))+r\left(\mu_{T_2}(x_{T_2}(\tau),\tau)\right)\Big)\textrm{d}\tau\\\nonumber&=\int_{0}^{T_1}\Big(Q(x_{T_2}(\tau))+r\left(\mu_{T_2}(x_{T_2}(\tau),\tau)\right)\Big)\textrm{d}\tau+\int_{T_1}^{T_2}\Big(Q(x_{T_2}(\tau))+r\left(\mu_{T_2}(x_{T_2}(\tau),\tau)\right)\Big)\textrm{d}\tau \\ &\ge \int_{0}^{T_1}\Big(Q(x_{T_2}(\tau))+r\left(\mu_{T_2}(x_{T_2}(\tau),\tau)\right)\Big)\textrm{d}\tau.\label{eq:T2T1a}
\end{align}
However, since $\mu_{T_1}$ minimizes the cost $J_{T_1}(\cdot;~x_0)$:
\begin{equation}\label{eq:T2T1b}
\int_{0}^{T_1}\Big(Q(x_{T_2}(\tau))+r\left(\mu_{T_2}(x_{T_2}(\tau),\tau)\right)\Big)\textrm{d}\tau=J_{T_1}(\mu_{T_2};~x)\ge J_{T_1}(\mu_{T_1};~x)=V_{T_1}(x,0).
\end{equation}
Hence, combining \eqref{eq:T2T1a}-\eqref{eq:T2T1b}, we obtain that $V_{T_2}(x_0,0)\ge V_{T_1}(x_0,0)$, for all $x_0\in \mathbb{R}^n$. The fact that $V_{T}(x_0,0)\le V^\star(x_0)$ for all $x_0\in\mathbb{R}^n$ and $T>0$ can be proved using an identical procedure, by setting $T_1=T$ while taking the limit $T_2\rightarrow\infty$. 

To conclude, note that for all $x\in\mathbb{R}^n$, the sequence $V_T(x,0)$ is increasing with respect to $T$ and bounded above by $V^\star(x)$. Therefore, $\lim_{T\rightarrow\infty}V_T(x,0)=L(x)$ for some Lebesgue measurable function $L:\mathbb{R}^n\rightarrow\mathbb{R}$,  $\forall x\in\mathbb{R}^n$. 
\end{proof}

\begin{lemma}\label{le:xt}
Assume that $\phi\equiv0$. Then, it holds that $\lim_{T\rightarrow\infty}x_T(T)=0$.
\end{lemma}
\begin{proof}
For each $T>0$, we have:
\begin{align}
\nonumber V_{T}(x_0,0)=J_{T}(x_0,\mu_{T})&=\int_{0}^{T}\Big(Q(x_{T}(\tau))+r\left(\mu_{T}(x_{T}(\tau),\tau)\right)\Big)\textrm{d}\tau\\\nonumber &=\int_{0}^{\frac{T}{2}}\Big(Q(x_{T}(\tau))+r\left(\mu_{T}(x_{T}(\tau),\tau)\right)\Big)\textrm{d}\tau +\int_{\frac{T}{2}}^{T}\Big(Q(x_{T}(\tau))+r\left(\mu_{T}(x_{T}(\tau),\tau)\right)\Big)\textrm{d}\tau.
\end{align}
Hence
\begin{align}\nonumber
V_{T}(x_0,0) &= J_\frac{T}{2}(x_0,\mu_{T})+\int_{\frac{T}{2}}^{T}\Big(Q(x_{T}(\tau))+r\left(\mu_{T}(x_{T}(\tau),\tau)\right)\Big)\textrm{d}\tau\nonumber\\\nonumber &\ge  J_\frac{T}{2}(x_0,\mu_{\frac{T}{2}})+\int_{\frac{T}{2}}^{T}\Big(Q(x_{T}(\tau))+r\left(\mu_{T}(x_{T}(\tau),\tau)\right)\Big)\textrm{d}\tau\\ &=V_{\frac{T}{2}}(x_0,0)+\int_{\frac{T}{2}}^{T}\Big(Q(x_{T}(\tau))+r\left(\mu_{T}(x_{T}(\tau),\tau)\right)\Big)\textrm{d}\tau.\label{eq:V2T}
\end{align}
Note that $\lim_{T\rightarrow\infty}V_\frac{T}{2}(x_0,0)=\lim_{T\rightarrow\infty}V_{T}(x_0,0)=L(x_0)$ according to Lemma \ref{le:inc}. Applying this to \eqref{eq:V2T} yields:
\begin{align*}
&\lim_{T\rightarrow\infty}\int_{\frac{T}{2}}^{T}\Big(Q(x_{T}(\tau))+r\left(\mu_{T}(x_{T}(\tau),\tau)\right)\Big)\textrm{d}\tau\le 0   \Longrightarrow \lim_{T\rightarrow\infty} \int_{\frac{T}{2}}^{T}\Big(Q(x_{T}(\tau))+r\left(\mu_{T}(x_{T}(\tau))\right)\Big)\textrm{d}\tau=0.
\end{align*}
Since $Q(\cdot)$ is positive definite and $R\succ0$, the energy of the state and the control converge to zero in the steady state as $T\rightarrow\infty$, which may hold only if $\lim_{T\rightarrow\infty} x_{T}(T)= 0$. 
\end{proof}

The two results above can be exploited to prove that, when $\phi\equiv0$, $V_T(\cdot,0)$ \textit{pointwise} approximates the optimal value function $V^\star(\cdot)$ as $T$ increases.

\begin{lemma}\label{le:point}
Assume that $\phi\equiv0$. Then, for each fixed $x\in\mathbb{R}^n$:
\begin{align*}
&\lim_{T\rightarrow\infty }V_T(x,0)=V^\star(x).
\end{align*}
\end{lemma}
\begin{proof}
Taking the time derivative of $V^\star$ over the trajectories generated by the control input $\mu_T$, we have
\begin{equation}
\begin{split}\label{eq:dotVs}
\dot{V}^\star(x)=&\nabla_x V^{\star\textrm{T}}(x)(f(x)+g(x)\mu_T(x,t))\\=&\nabla_x V^{\star\textrm{T}}(f(x)+g(x)\mu^\star(x))+\nabla_x V^{\star\textrm{T}} (x)g(x)(\mu_T(x,t)-\mu^\star(x)).
\end{split}
\end{equation}
From \eqref{eq:mustar}-\eqref{eq:HJB}, we obtain the following relations:
\begin{align*}
&\nabla_x V^{\star\textrm{T}}(f(x)+g(x)\mu^\star(x))=-Q(x)-\mu^{\star \textrm{T}}(x)R\mu^\star(x),\\
&\nabla_x V^{\star\textrm{T}}(x) g(x)=-2\mu^{\star\textrm{T}}(x)R.
\end{align*}
Combining these with \eqref{eq:dotVs} yields
\begin{equation*}
\dot{V}^\star(x)=-Q(x)-\mu^{\star \textrm{T}}(x)R\mu^\star(x)-2\mu^{\star\textrm{T}}(x)R(\mu_T(x,t)-\mu^\star(x)),
\end{equation*}
or equivalently
\begin{align*}
\dot{V}^\star(x)=&-Q(x)-\mu_T^{\textrm{T}}(x,t)R\mu_T(x,t)+(\mu^\star(x)-\mu_T(x,t))^\textrm{T}R(\mu^\star(x)-\mu_T(x,t)).
\end{align*}
Integrating this equation over $t\in[0,T]$, and recalling that the trajectories of $x$ are generated by the control $\mu_T$, we get:
\begin{align}\label{eq:dotVs2}
&V^\star(x_T(T))-V^\star(x_0)=-\int_0^T(Q(x_T)+\mu_T^{\textrm{T}}(x_T,t)R\mu_T(x_T,t))\textrm{d}t+\int_0^T(\mu^\star(x_T)-\mu_T(x_T,t))^\textrm{T}R(\mu^\star(x_T)-\mu_T(x_T,t))\textrm{d}t.
\end{align}
Recalling again that the integrals are taken over the trajectories generated by $\mu_T$, and since $\phi\equiv0$, we have by \eqref{eq:costT} and \eqref{eq:dotVs2} that:
\begin{equation*}
V^\star(x_T(T))-V^\star(x_0)=-V_T(x_0,0)+\int_0^T(\mu^\star(x_T)-\mu_T(x_T,t))^\textrm{T}R(\mu^\star(x_T)-\mu_T(x_T,t))\textrm{d}t
\end{equation*}
or equivalently
\begin{align}\label{eq:dotVs3}
&V_T(x_0,0)=V^\star(x_0)-V^\star(x_T(T))+\int_0^T(\mu^\star(x_T)-\mu_T(x_T,t))^\textrm{T}R(\mu^\star(x_T)-\mu_T(x_T,t))\textrm{d}t.
\end{align}
Since $R\succ0$, this implies that $V_T(x_0,0)\ge V^\star(x_0)-V^\star(x_T(T))$. In addition, from Lemma \ref{le:inc} we have $V_T(x_0,0) \le V^\star(x_0)$, hence overall:
\begin{equation}\label{eq:ineq2}
V^\star(x_0)-V^\star(x_T(T))\le V_T(x_0,0) \le V^\star(x_0).
\end{equation}
It holds by continuity and Lemma \ref{le:xt} that $\lim_{T\rightarrow\infty}V^\star(x_T(T))=V^\star(\lim_{T\rightarrow\infty}x_T(T))=V^\star(0)=0$. Therefore, 
taking the limit as $T\rightarrow\infty$ in \eqref{eq:ineq2}, we obtain:
\begin{equation*}
V^\star(x_0)\le \lim_{T\rightarrow\infty} V_T(x_0,0) \le V^\star(x_0),
\end{equation*}
i.e. $\lim_{T\rightarrow\infty} V_T(x_0,0) = V^\star(x_0).$
\end{proof}
Finally, the uniform convergence of $V_T(\cdot,0)$ to $V^\star$ under $\phi=0$ follows by combining Lemmas \ref{le:inc} and \ref{le:point}.
\begin{lemma}\label{le:ucon1}
Assume that $\phi\equiv0$. Then, $V_T(\cdot,0)$ converges uniformly to $V^\star$ on $\Omega$.
\end{lemma}
\begin{proof}
From Lemmas \ref{le:inc} and \ref{le:point}, the sequence of continuous functions $V_T(\cdot,0)$ is increasing and converges pointwise everywhere on $\Omega$ to the continuous function $V^\star$. Since $\Omega$ is compact, it follows from Dini's theorem  (Theorem 7.13 in \cite{rudin1976principles})  that the mode of convergence is uniform.
\end{proof}

\subsection{Uniform Approximation with $\phi\equiv\psi+V^\star$}

In this subsection, we prove a more straightforward result: that for $\phi\equiv\psi+V^\star$, where $\psi$ is positive semidefinite, $V_T(\cdot,0)$ once again converges uniformly to $V^\star$ on $\Omega$.

\begin{lemma}\label{le:ucon2}
Assume that $\phi\equiv\psi+V^\star$ for some positive semidefinite, continuous function $\psi$. Then, $V_T(\cdot,0)$ converges uniformly to $V^\star$ on $\Omega$.
\end{lemma}
\begin{proof}
Under this terminal cost, we have from \eqref{eq:costT}:
\begin{equation}\label{eq:Jt1}
    V_T(x_0, 0)\le J_T(\mu^\star;~x_0)=\int_{0}^{T}\Big(Q(x^\star(\tau))+r\left(\mu^\star(x^\star(\tau))\right)\Big)\textrm{d}\tau+V^\star(x^\star(T))+\psi(x^\star(T))=V^\star(x_0)+\psi(x^\star(T)).
\end{equation}
On the other hand:
\begin{equation}\nonumber
    J_T(\mu;~x_0)=\int_{0}^{T}\Big(Q(x(\tau))+r\left(\mu(x(\tau),\tau)\right)\Big)\textrm{d}\tau+\psi(x(T))+V^\star(x(T))\ge \int_{0}^{T}\Big(Q(x(\tau))+r\left(\mu(x(\tau),\tau)\right)\Big)\textrm{d}\tau+V^\star(x(T)).
\end{equation}
By the principle of optimality, the right-hand side is minimized when $\mu=\mu^\star$. Hence:
\begin{equation}\label{eq:Jt2}
    J_T(\mu;~x_0)\ge \int_{0}^{T}\Big(Q(x^\star(\tau))+r\left(\mu^\star(x^\star(\tau))\right)\Big)\textrm{d}\tau+V^\star(x^\star(T))=V^\star(x_0).
\end{equation}
Finally, setting $\mu=\mu_T$, one has $J_T(\mu_T;~x_0)=V_T(x_0,0)$, hence \eqref{eq:Jt2} gives 
\begin{equation}\label{eq:Jt3}
V_T(x_0,0)\ge V^\star(x_0).
\end{equation}
 Combining \eqref{eq:Jt1} with \eqref{eq:Jt3}, we obtain:
\begin{multline}\label{eq:final}
V^\star(x_0)\le V_T(x_0,0)\le V^\star(x_0)+\psi(x^\star(T)) \Longrightarrow |V_T(x_0,0)-V^\star(x_0)|\le \psi(x^\star(T)) \\\Longrightarrow \max_{x_0\in\Omega}|V_T(x_0,0)-V^\star(x_0)|\le \max_{x_0\in\Omega}\psi(x^\star(T)).
\end{multline}
But $\mu^\star$ is asymptotically stabilizing on $\Omega$, hence $\lim_{T\rightarrow\infty}\max_{x_0\in\Omega}\psi(x^\star(T))=0$. Owing to \eqref{eq:final}, this implies $\lim_{T\rightarrow\infty}\max_{x_0\in\Omega}|V_T(x_0,0)-V^\star(x_0)|=0$, i.e., $V_T(\cdot,0)$ converges uniformly to $V^\star$ on $\Omega$.
\end{proof}

\subsection{Uniform Approximation with Arbitrary Terminal Cost}

To conclude with the main result of this section, recall that \eqref{eq:sandwich} implies that if $V_T(\cdot,0)$ converges uniformly to $V^\star$ as $T\rightarrow\infty$ for $\phi\equiv0$ and $\phi\equiv\psi+V^\star$, then the same result is generalizable for any arbitrary positive semi-definite terminal cost $\phi$. Hence, we get the following theorem.
\begin{theorem}\label{th:ucon}
    Let $\phi$ be an arbitrary, positive semi-definite terminal cost for \eqref{eq:costT}. Then, 
    $V_T(\cdot,0)\rightarrow V^\star$ and $\mu_T(\cdot,0)\rightarrow\mu^\star$ uniformly on $\Omega$.
\end{theorem}
\begin{proof}
The fact that $V_T(\cdot,0)\rightarrow V^\star$ uniformly on $\Omega$ follows directly from \eqref{eq:sandwich} and Lemmas \ref{le:ucon1}-\ref{le:ucon2}. Moreover, following the same arguments as in \cite{abu2005nearly},  the Lipschitz continuity of \eqref{eq:sys} as well as the continuity of the controllers $\mu_T$ imply that, as $T\rightarrow\infty$, the trajectories of \eqref{eq:sys} under $u=\mu_T$ converge uniformly $\forall x_0\in\Omega$, hence $\mu_T(\cdot,0)\rightarrow\mu^\star$ uniformly on $\Omega$ since $V_T(\cdot,0)\rightarrow V^\star$ uniformly.
\end{proof}

 \begin{remark}
Combining Theorem \ref{th:ucon} with the universal approximation property of neural networks, we conclude that approximating the finite-horizon value function $V_T$ with PINNs for a large horizon $T$ is a valid procedure towards estimating the infinite horizon value function $V^\star$ over a compact set.
 \end{remark}

 Note that \eqref{eq:sandwich} implies $V_T(\cdot,0; 0)-V^\star(\cdot)\le V_T(\cdot,0; \psi)-V^\star(\cdot) \le V_T(\cdot,0; \psi+V^\star)-V^\star(\cdot)$. Combining this with \eqref{eq:ineq2} and \eqref{eq:final}, it follows that
 \begin{equation*}
\max_{x_0\in\Omega}|V_T(x_0,0)-V^\star(x_0)|\le \max\{ \max_{x_0\in\Omega}\psi(x^\star(T)), ~\max_{x_0\in\Omega} V^\star(x_T(T;~0)) \},
 \end{equation*}
 where $x_T(\cdot;~0)$ denotes the trajectory of $x_T$ when $\phi\equiv0$. 
 The expression above gives us some information regarding the quality of the approximation of $V^\star(\cdot)$ by $V_T(\cdot,0)$ with respect to $T$. Specifically, it tells us that the approximation quality depends on how quickly the optimal trajectories $x^\star$ and $x_T$ at time $T$ converge to zero as $T$ increases, which is natural as the convergence rates of  $x^\star$ and $x_T$ dictate how quickly the tail of the integral \eqref{eq:cost} converges to 0.

\section{Effect of the Horizon Length $T$}\label{sec:horizon}

\subsection{Evaluation of the Horizon Length}

So far, it has been proven that the finite-horizon value function $V_T(\cdot,0)$ indeed uniformly approximates the infinite-horizon one, namely $V^\star$, as $T\rightarrow\infty$. Therefore, using PINNs to get an estimate of $V_T(\cdot,0)$ is, in principle, a valid procedure in the task of obtaining an estimate of $V^\star$, provided that $T$ is large enough. This brings forward another question: how large should $T$ be chosen? Or, in other words, how can we evaluate whether the chosen horizon $T$ is large enough, so that $\mu_T(\cdot,0)$, $V_T(\cdot,0)$ are close to  $\mu^\star$, $V^\star$?

One way to verify whether $V_T(\cdot,0)$ is close enough to $V^\star$ in $\Omega$, is by checking whether it satisfies the infinite-horizon HJB \eqref{eq:HJB}, since $V^\star$ is in fact a solution to that equation. Towards this end, one can define the following flow residual:
\begin{equation}\label{eq:flow_res}
E_e(x; V_T(\cdot,0))=\nabla_x V_T^{\textrm{T}}(x,0)f(x)+Q(x)-\frac{1}{4}\nabla_x V_T^{\textrm{T}}(x,0)g(x)R^{-1}g^\textrm{T}(x)\nabla_x V_T(x,0).
\end{equation}
This is essentially an error indicating how far $V_T(\cdot,0)$ is from satisfying the infinite-horizon HJB \eqref{eq:HJB} at the point $x$; it is zero if and only if $V_T(\cdot,0)$ solves \eqref{eq:HJB} at this specific point. Accordingly, one can define the boundary residual:
\begin{equation*}
E_b(V_T(\cdot,0))=V_T(0,0),
\end{equation*}
which is zero if and only if $V_T(0,0)=0$, i.e. if and only if $V_T(\cdot,0)$ satisfies the boundary condition of the infinite-horizon HJB \eqref{eq:HJB}. Therefore, by aggregating the residuals into a single error term:
\begin{equation}\label{eq:IHJB_res}
E=\frac{1}{N_c}\sum_{i=1}^{N_c}E_e(x^i_c; V_T(\cdot,0))^2+E_b(V_T(\cdot,0))^2,
\end{equation}
where $x^i_c\in\Omega$, $i=1,\ldots,N_c$, are data points, a procedure to check whether $V_T(\cdot,0)$ provides a good approximation to $V^\star$ would be to check how close $E$ is to zero.
\begin{remark}
Since PINNs are used to approximate $V_T$, only an approximation $\hat{V}_T$ of $V_T$ is available. Therefore, in practice, one would only be able to compute the residuals $E=\frac{1}{N_c}\sum_{i=1}^{N_c}E_e(x^i_e; \hat{V}_T(\cdot,0))^2+E_b(\hat{V}_T(\cdot,0))^2$, and check whether those are close to zero or not. Since $\hat{V}_T$ would only approximate $V_T$, a nonzero residual may not necessarily imply that $T$ is not large enough, but it could mean that the underlying neural network architecture is not expressive enough to sufficiently approximate $V_T$.
\end{remark}

\subsection{Extension of the Horizon Length}\label{sub:ext}

In case that the horizon $T$ is deemed to be small, by means of the residual error $E$ being large, then one can increase this horizon to get a better approximation of $V^\star$. In that respect, one could recompute $V_{T'}$ from scratch for some larger horizon $T'>T$ by reemploying the PINNs method of Section \ref{sec:PINNs} over the larger domain $\Omega\times[0,T']\supset\Omega\times[0,T]$. However, this can be a tedious procedure, and in fact unnecessary; instead, the value function $V_T$ that has already been computed for a smaller horizon can be used as an aid to find $V_{T'}$, $T'>T$, without resolving the whole problem from scratch over $[0,T']$, but rather by solving another smaller problem only over $[0,T'-T]$. To see this, consider the finite-horizon cost functional, which is similar to \eqref{eq:costT} but augmented with a different terminal cost, dependent on $V_T$:
\begin{equation}\label{eq:costinc}
    J'(\mu; x_0)=\int_{0}^{T'-T}\Big(Q(x(\tau))+r\left(\mu(x(\tau),\tau)\right)\Big)\textrm{d}\tau+V_T(x(T'-T),0).
\end{equation}
The optimal value function $V':\mathbb{R}^n\times [0,T'-T]\rightarrow\mathbb{R}$ and control $\mu':\mathbb{R}^n\times [0,T'-T]\rightarrow\mathbb{R}^m$ of this problem satisfy:
\begin{align}\label{eq:HJBinc}
&\nabla_t V^{'}(x,t)+\nabla_x V^{'\textrm{T}}(x,t)f(x)+Q(x)-\frac{1}{4}\nabla_x V^{'\textrm{T}}(x,t)g(x)R^{-1}g^\textrm{T}(x)\nabla_x V'(x,t)=0,\\& V'(x,T'-T)=V_T(x,0),\nonumber
\end{align}
and
\begin{equation}\label{eq:muinc}
\mu'(x,t)=-\frac{1}{2}R^{-1}g^\textrm{T}(x)\nabla_x {V'}(x,t).
\end{equation}
Notice that these two functions are defined only over $t\in[0,T'-T]$ and not over $t\in[0,T']$, and the same holds for the corresponding PDE \eqref{eq:HJBinc}. In what follows, we show that, in fact, $V'(\cdot,0)=V_{T'}(\cdot,0)$ and $\mu'(\cdot,0)=\mu_{T'}(\cdot,0)$.

\begin{theorem}\label{th:inc}
It holds that $V'(\cdot,0)=V_{T'}(\cdot,0)$, and $\mu'(\cdot,0)=\mu_{T'}(\cdot,0)$.
\end{theorem}
\begin{proof}
The value function $V_{T'}(x_0,0)$, by definition, is equal to the optimal value of the following cost functional for any $x_0\in\mathbb{R}^n$:
\begin{align*}
    J_{T'}(\mu;~x_0)=&\int_{0}^{T'}\Big(Q(x(\tau))+r\left(\mu(x(\tau),\tau)\right)\Big)\textrm{d}\tau+\phi(x(T'))\\=&\int_{0}^{T'-T}\Big(Q(x(\tau))+r\left(\mu(x(\tau),\tau)\right)\Big)\textrm{d}\tau+\int_{T'-T}^{T'}\Big(Q(x(\tau))+r\left(\mu(x(\tau),\tau)\right)\Big)\textrm{d}\tau+\phi(x(T')).
\end{align*}
However, by the principle of optimality, we have:
\begin{align*}
V_{T'}(x_0,0)&=J_{T'}(\mu_{T'};~x_0)=\min_\mu J_{T'}(\mu;~x_0)\\&=\min_{\mu}\Big(\int_{0}^{T'-T}\Big(Q(x(\tau))+r\left(\mu(x(\tau),\tau)\right)\Big)\textrm{d}\tau+\int_{T'-T}^{T'}\Big(Q(x(\tau))+r\left(\mu(x(\tau),\tau)\right)\Big)\textrm{d}\tau+\phi(x(T'))\Big)
\\&=\min_{\mu}\Big(\int_{0}^{T'-T}\Big(Q(x(\tau))+r\left(\mu(x(\tau),\tau)\right)\Big)\textrm{d}\tau+\min_{\mu'}\Big(\int_{T'-T}^{T'}\Big(Q(x(\tau))+r\left(\mu'(x(\tau),\tau)\right)\Big)\textrm{d}\tau+\phi(x(T'))\Big)\Big)\\&=\min_{\mu}\Big(\int_{0}^{T'-T}\Big(Q(x(\tau))+r\left(\mu(x(\tau),\tau)\right)\Big)\textrm{d}\tau+V_T(x(T'-T),0)\Big)\\&=\min_\mu J'(\mu;~x_0)=V'(x_0,0),
\end{align*}
which is the required result. The fact that $\mu'(\cdot,0)=\mu_{T'}(\cdot,0)$ then follows from \eqref{eq:muTstar} and \eqref{eq:muinc} and the fact that $V_{T'}(x_0,0)=V'(x_0,0)$ for all $x_0\in\mathbb{R}^n$.
\end{proof}

Based on Theorem \ref{th:inc}, if one knows the value function $V_T(\cdot,0)$ for some horizon length $T>0$, then they can compute the value function $V_{T'}(\cdot,0)$ over a larger horizon $T'>T$ by solving the PDE \eqref{eq:HJBinc}. This PDE is defined only over a time length of $T'-T$, so it is significantly less tedious to deal with than resolving the PDE \eqref{eq:HJBT} from scratch. In that respect, if one needs to increase the horizon length $T$ so that $V_T(\cdot,0)$ better approximates $V^\star$, it would be less tedious to apply the PINN method on \eqref{eq:HJBinc}, instead applying it from scratch on \eqref{eq:HJBT} for a larger horizon.

\begin{remark}
The PDE \eqref{eq:HJBinc} can be approximately solved similarly to \eqref{eq:HJBT}, by sampling data points on the flow and the boundary condition and defining a loss function of the form \eqref{eq:MSE}. The specific details are straightforward and omitted to avoid repetition.
\end{remark}

With regard to the computational complexity of the horizon extension mechanism, there is a tradeoff between i) choosing a new horizon length $T'$ that is large enough and ii) using a neural network structure that is relatively small and thus faster to train. On the one hand, choosing $T'$ to be much larger than $T$ allows one to achieve a desired tolerance one-shot, without having to re-extend the horizon length multiple times. On the other hand, choosing $T'$ to be only slightly larger than $T$ means that the horizon-extension problem \eqref{eq:costinc} can be solved with a smaller neural network, because the time domain of the extension lies on $[T, T']$. Similarly, choosing an initial horizon length $T$ that is large enough can enable one-shot learning of the infinite-horizon optimal controller without resorting to any iterations. Nevertheless, a small horizon $T$ combined with iterative extensions can enable the use of physics-informed neural networks that are smaller in size and thus faster to train.

\subsection{Effect of Approximation Errors on Horizon Length Extension}

Typically, the finite-horizon value function $V_T(\cdot,0)$ cannot be computed exactly owing to the finite approximation capabilities of PINNs. This is not an issue if $T$ is chosen to be sufficiently large from the very beginning, however, problems may arise if $T$ is initially small and then extended using the method of Subsection \ref{sub:ext}. In this case, errors in the estimation of $V_T(\cdot,0)$ owed to the finite approximation capabilities of the PINNs may be carried over to the estimation of $V_{T'}(\cdot,0)$, and such errors can accumulate if the horizon is extended multiple times.

To study the effect of approximation errors on the extension of the horizon length $T$, suppose that instead of $V_T(\cdot,0)$, the estimated finite-horizon value function given by PINNs is $\hat{V}_T(\cdot,0)$, and denote as $e_T(x)=\hat{V}_T(x,0)-V_T(x,0)$ the approximation error for any $x\in\Omega$. In this case, if we apply the method of Subsection \ref{sub:ext} to extend the horizon length to $T'>T$, instead of $V_{T'}(\cdot,0)$ we will obtain a value function estimate $\hat{V}_{T'}(x,0)$ that is the minimum value of the cost:
\begin{equation}\label{eq:costinc2}
    \hat{J}'(\mu; x_0)=\int_{0}^{T'-T}\Big(Q(x(\tau))+r\left(\mu(x(\tau),\tau)\right)\Big)\textrm{d}\tau+\hat{V}_T(x(T'-T),0).
\end{equation}
Note that here, unlike \eqref{eq:costinc}, $\hat{V}_T$ is used as a terminal cost because the exact $V_T$ is unknown owing to approximation errors. Accordingly, instead of $\mu_{T'}(\cdot,0)$, we will obtain an optimal control estimate $\hat{\mu}_{T'}(x,0)$ that is the minimizer of \eqref{eq:costinc2} at $t=0$.

In the following theorem, we provide some intuition on how close $\hat{V}_{T'}(\cdot,0)$ will be to ${V}_{T'}(\cdot,0)$, based on the approximation error $e_T=\hat{V}_T(x,0)-V_T(x,0)$.
\begin{theorem}\label{th:error_bound}
Denote as $\hat{x}_{T'}$ the trajectories of \eqref{eq:sys} under $u=\hat{\mu}_{T'}$. Then, for all $x_0\in \Omega$:
\begin{equation}\label{eq:error_bound}
|\hat{V}_{T'}(x_0,0)-V_{T'}(x_0,0)|\le \max\{|e_T(x_{T'}(T'-T))|, |e_T(\hat{x}_{T'}(T'-T))|\}.
\end{equation}
\end{theorem}
\begin{proof}
We have:
\begin{equation}\label{eq:error1}
\begin{split}
\hat{V}_{T'}(x_0,0)&=\min_{\mu} \hat{J}'(\mu;~x_0)\le  \hat{J}'(\mu_{T'};~x_0)\\&=\int_{0}^{T'-T}\Big(Q(x_{T'}(\tau))+r\left(\mu_{T'}(x_{T'}(\tau),\tau)\right)\Big)\textrm{d}\tau+\hat{V}_T(x_{T'}(T'-T),0)\\&=\int_{0}^{T'-T}\Big(Q(x_{T'}(\tau))+r\left(\mu_{T'}(x_{T'}(\tau),\tau)\right)\Big)\textrm{d}\tau+{V}_T(x_{T'}(T'-T),0)+e_T(x_{T'}(T'-T))\\&={V}_{T'}(x_0,0)+e_T(x_{T'}(T'-T)).
\end{split}
\end{equation}
Moreover:
\begin{equation}\label{eq:error2}
\begin{split}
{V}_{T'}(x_0,0)&=\min_{\mu} {J}'(\mu;~x_0)\le  {J}'(\hat{\mu}_{T'};~x_0)\\&=\int_{0}^{T'-T}\Big(Q(\hat{x}_{T'}(\tau))+r\left(\hat{\mu}_{T'}(\hat{x}_{T'}(\tau),\tau)\right)\Big)\textrm{d}\tau+{V}_T(\hat{x}_{T'}(T'-T),0)\\&=\int_{0}^{T'-T}\Big(Q(\hat{x}_{T'}(\tau))+r\left(\hat{\mu}_{T'}(\hat{x}_{T'}(\tau),\tau)\right)\Big)\textrm{d}\tau+\hat{V}_T(\hat{x}_{T'}(T'-T),0)-e_T(\hat{x}_{T'}(T'-T))\\&=\hat{V}_{T'}(x_0,0)-e_T(\hat{x}_{T'}(T'-T)).
\end{split}
\end{equation}
Combining \eqref{eq:error1}-\eqref{eq:error2}, we obtain
\begin{equation*}
e_T(\hat{x}_{T'}(T'-T))\le \hat{V}_{T'}(x_0,0)-V_{T'}(x_0,0)\le e_T(x_{T'}(T'-T)).
\end{equation*}
This yields \eqref{eq:error_bound}.
\end{proof}
Theorem \ref{th:error_bound} provides some intuition on how approximation errors in the estimation of the value function $V_{T}$ affect the horizon-extending procedure of Subsection \ref{sub:ext}. In particular, if the approximation error $e_T$ is bounded by a constant $\epsilon>0$, then Theorem \ref{th:error_bound} implies the following error bound for the extended-horizon value function estimate $\hat{V}_{T'}$:
\begin{equation}\label{eq:loose_bound}
|\hat{V}_{T'}(x_0,0)-V_{T'}(x_0,0)|\le \epsilon.
\end{equation}

\begin{figure}[!t]
	\begin{center}
		\includegraphics[width=14cm, height=5.3cm]{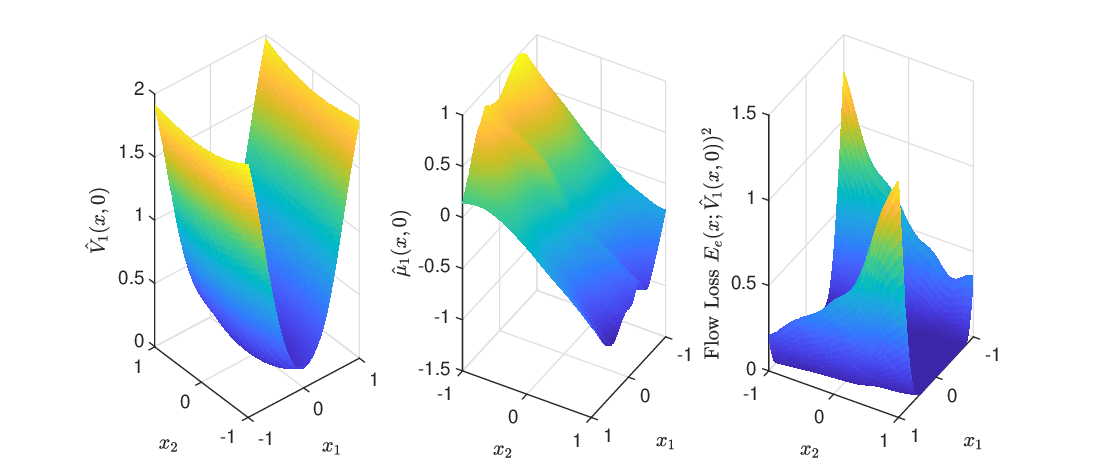}
		\caption{The learnt value function $\hat{V}_{1}(\cdot,0)$ (left), the learnt control policy $\hat{\mu}_{1}(\cdot,0)$ (middle), and the squared flow residual $E_e(x;\hat{V}_{1}(x,0))^2$ (right), for $T=1$ for the pendulum.}	
		\label{fig:pend1}
	\end{center}
\end{figure}

\begin{figure}[!t]
	\begin{center}
		\includegraphics[width=14cm, height=5.3cm]{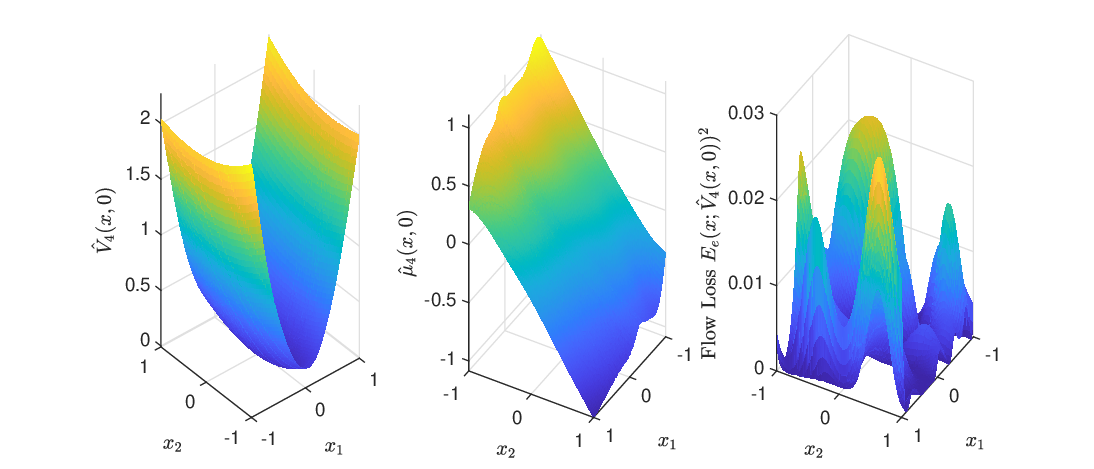}
		\caption{The learnt value function $\hat{V}_{4}(\cdot,0)$ (left), the learnt control policy $\hat{\mu}_{4}(\cdot,0)$ (middle), and the squared flow residual $E_e(x;\hat{V}_{4}(x,0))^2$ (right), for the pendulum.}	
		\label{fig:pend4}
	\end{center}
\end{figure}
\begin{figure}[!t]
	\begin{center}
		\includegraphics[width=14cm, height=5.3cm]{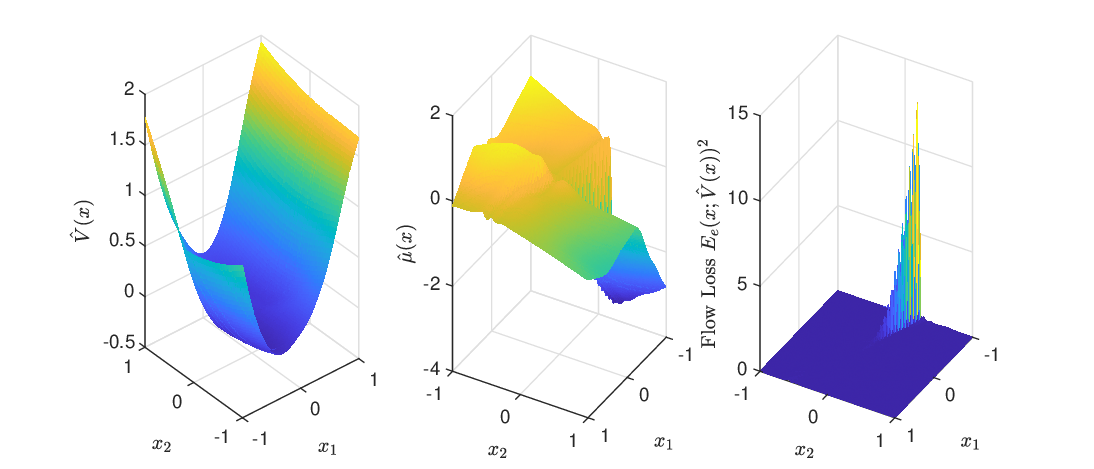}
		\caption{The learnt value function $\hat{V}(\cdot)$ (left), the learnt control policy $\hat{\mu}(\cdot)$ (middle), and the squared flow residual $E_e(x;\hat{V}(x))^2$ (right), for the pendulum. }	
		\label{fig:pendinf}
	\end{center}
\end{figure}

Yet, the result of the Theorem is much stricter than \eqref{eq:loose_bound}. In particular, note that by Lemma \ref{le:xt}, one has $\lim_{T'\rightarrow\infty}x_{T'}(T'-T)=0$, therefore the continuity of the approximation error implies $\lim_{T'\rightarrow\infty}e_T(x_{T'}(T'-T))=e_T(0)$. Hence, if the basis functions of the PINN are chosen so that $\hat{V}_T(0,0)=0$, or if this condition is enforced through the residual $F_{in}$, then $e_T(0)=\hat{V}_T(0,0)-{V}_T(0,0)=0$, and thus $\lim_{T'\rightarrow\infty}e_T(x_{T'}(T'-T))=e_T(0)=0$.
Accordingly, if the extended horizon length $T'$ is large enough, then this implies that $|e_T(x_{T'}(T'-T))|$ will be small. Moreover, as the approximation capabilities of the neural network increase, $\hat{x}_{T'}$ converges uniformly to $x_{T'}$, thus $e_T(\hat{x}_{T'}(T'-T))$ also becomes smaller as the extended horizon length $T'$ increases. Combining these two observations and looking at equation \eqref{eq:error_bound}, we conclude a much stronger result than \eqref{eq:loose_bound}: that a large enough horizon extension rules out the phenomenon of ``accumulating" approximation errors when the method of Subsection \ref{sub:ext} is applied multiple times.

\section{Simulations}\label{sec:sim}

\subsection{Torsonial Pendulum}

We first perform simulations on a torsonial pendulum \cite{liu2013policy}, with dynamics: 
\begin{align*}
\dot{\theta}&=\omega,\\
J\dot{\omega}&=u-Mgl\textrm{sin}(\theta)-f_d\omega,
\end{align*}
where $M = 1/3$ kg and $l = 2/3$ m are the mass and length
of the pendulum bar, respectively. In addition, $J = 4/3Ml^2$
and $f_d=0.2$ are the rotary inertia and frictional factor, respectively, and $g=9.8~\textrm{m/s}^2$. Choosing $x_1=\theta$ and $x_2=\omega$, the system can be brought to the form \eqref{eq:sys} with $f(x)=\begin{bmatrix}x_2 & (-Mgl\textrm{sin}(x_1)-f_dx_2)/J \end{bmatrix}^\textrm{T}$ and $g(x)=1/J$. For this system, we want to solve the infinite-horizon optimal control problem under cost \eqref{eq:cost} with $R=1$ and $Q(x)=x_1^2+x_2^2$.

To compute the optimal value function $V^\star$ for this system over the compact set $\Omega=[-1,1]^2$\footnote{In practice, the compact set $\Omega$ should be chosen so that it captures the regions of the state space that our system usually operates over. }, we utilize a 3-layer neural network with $100$ nodes in each layer, and with the activation function chosen to be the hyperbolic tangent, and apply the PINN method of Section \ref{sec:PINNs}. We train this network by defining the loss function as the residual \eqref{eq:MSE} of the finite-horizon HJB with horizon length $T=1$ and $\phi\equiv0$, with $N_e=10000$ flow data points, and $N_b=N_i=1600$ boundary data points. This loss function is then minimized using Adam \cite{kingma2014adam}, with an initial learning rate equal to $0.005$, decaying at a rate of $0.002$ each time a mini-batch of $500$ flow data points is used. The procedure terminates once the loss function becomes less than $10^{-3}$ for each mini-batch. To account for ill-conditioning at the corners of the state space, the set $\Omega$ during training is chosen to be equal to $[-1.5,1.5]\supset[-1,1]$.

The results are shown in Figure \ref{fig:pend1}, which depicts the estimated value function, optimal control, as well as the plot of the squared flow residual $E_e$ defined in \eqref{eq:flow_res}. We can see that while the estimated value function and optimal control are reasonable in terms of directions, the flow residual is still large, with a mean squared value of $0.045$. This implies that the horizon chosen was not large enough, so we extend it to $4$ seconds by using the method of Subsection \ref{sub:ext}, in three increments of $1$ second. The results are shown in Figure \ref{fig:pend4}, where clearly the residual is now much smaller, with a mean squared value of $0.0022$; more than ten times smaller than previously. Hence, the estimated value function and optimal control are now much more accurate.

\begin{table}[!t]\normalsize
  \begin{center}
  \caption{Mean-squared errors for different horizon lengths.}\label{tab:1}
   \begin{tabular}{c||c|c|c|c|c}
      \toprule 
      \textbf{Horizon Length $T$} & \textbf{$1$} & \textbf{$2$} & \textbf{$3$} & \textbf{$4$} & \textbf{$\infty$}\\
      \midrule 
      MSE \eqref{eq:IHJB_res} & $0.0453$ & $0.0041$ & $0.0026$ & $0.0022$ & $0.0122$\\
      \bottomrule 
    \end{tabular}
  \end{center}
\end{table}

We additionally perform a simulation in which we solve the infinite-horizon optimal control problem for the pendulum by directly applying PINNs to the steady-state HJB \eqref{eq:HJB}, as in \cite{furfaro2022physics}. All parameters are chosen as previously, and the results are shown in Figure \ref{fig:pendinf}. From this figure, we observe that the obtained approximation of the optimal value function and control is quite poor and, in fact, discontinuous, a property that could be attributed to the multiple solutions satisfying \eqref{eq:HJB}.  A quantitative comparison of the mean-squared error by directly solving the steady-state HJB \eqref{eq:HJB} as in \cite{furfaro2022physics}, denoted with $T=\infty$, and by using our approach for different horizon lengths, denoted with $T=1, 2, 3, 4$, is also shown in Table \ref{tab:1}. Clearly, our approach brought the mean-squared error below $10^{-2}$ for horizon lengths greater than $1$, whereas directly solving the steady-state HJB was unable to do so. These observations validate the motivation of this paper, which is that directly applying PINNs to the steady-state HJB can lead to issues of convergence.

\begin{figure}[!t]
	\begin{center}
		\includegraphics[width=10cm, height=6cm]{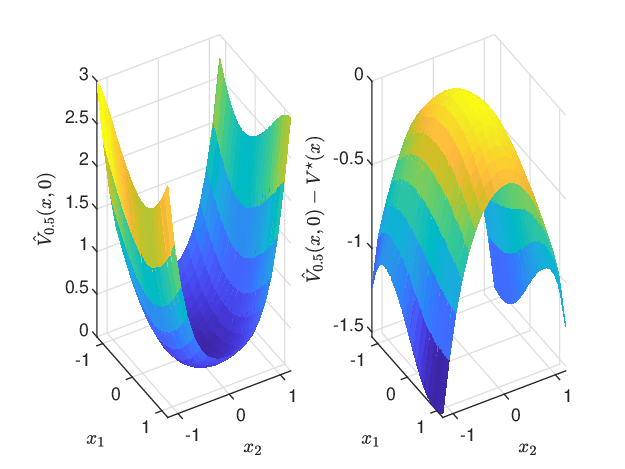}
		\caption{The learnt value function $\hat{V}_{0.5}(\cdot,0)$ (left) and the error from optimality $\hat{V}_{0.5}(\cdot,0)-V^\star(\cdot)$ (right), for $T=0.5$.}	
		\label{fig:V2}
	\end{center}
\end{figure}

\begin{figure}[!t]
	\begin{center}
		\includegraphics[width=10cm, height=6cm]{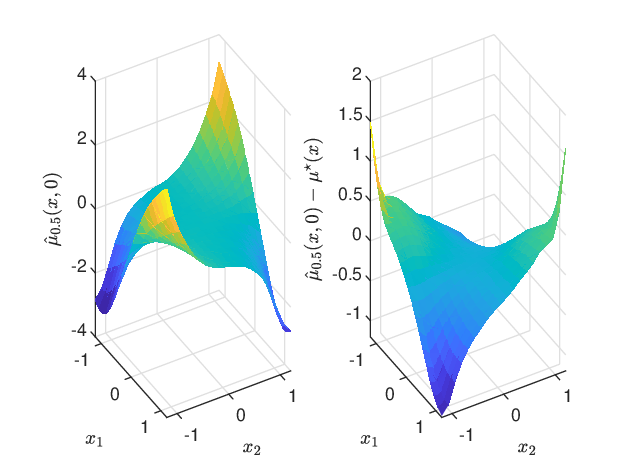}
		\caption{ The learnt control policy $\hat{\mu}_{0.5}(\cdot,0)$ (left) and the error from optimality $\hat{\mu}_{0.5}(\cdot,0)-\mu^\star(\cdot)$ (right), for $T=0.5$.}	
		\label{fig:u2}
	\end{center}
\end{figure}

\subsection{System with Quartic Value Function}

We now perform simulations on a more complex nonlinear system with a known quartic value function, borrowed from \cite{vrabie2009neural}:
\begin{align*}
\dot{x}_1&=-x_1+x_2+2x_2^3,\\
\dot{x}_2&=-\frac{1}{2}(x_1+x_2)+\frac{1}{2}x_2(1+2x_2^2)\textrm{sin}^2(x_1)+\textrm{sin}(x_1)u.
\end{align*}
This system can be written in form \eqref{eq:sys} if one selects $f(x)=\begin{bmatrix}-x_1+x_2+2x_2^3 & &  -\frac{1}{2}(x_1+x_2)+\frac{1}{2}x_2(1+2x_2^2)\textrm{sin}^2(x_1)\end{bmatrix}^\textrm{T}$ and $g(x)=\begin{bmatrix}0 & \textrm{sin}(x_1)\end{bmatrix}^\textrm{T}$. In addition, for this system, it is known that if $R=1$ and $Q(x)=x_1^2+x_2^2+x_2^4$, then the optimal value function entails a quartic term and is given by $V^\star(x)=\frac{1}{2}x_1^2+x_2^2+x_2^4$. Accordingly, from formula \eqref{eq:mustar}, the optimal control policy is $\mu^\star(x)=-\textrm{sin}(x_1)(x_2+2x_2^3)$.

\begin{figure}[!t] 
		\centering
                \includegraphics[width=10cm, height=6cm]{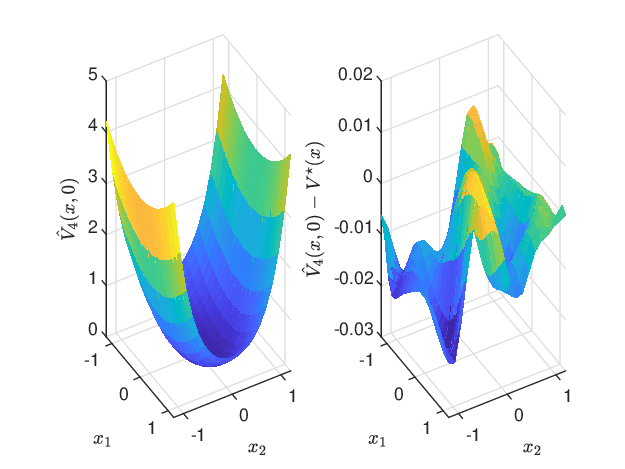}
                \caption{The learnt value function $\hat{V}_{4}(\cdot,0)$ (left) and the error from optimality $\hat{V}_{4}(\cdot,0)-V^\star(\cdot)$ (right).}
                \label{fig:V10}
\end{figure}

\begin{figure}[!t] 
		\centering
                \includegraphics[width=10cm, height=6cm]{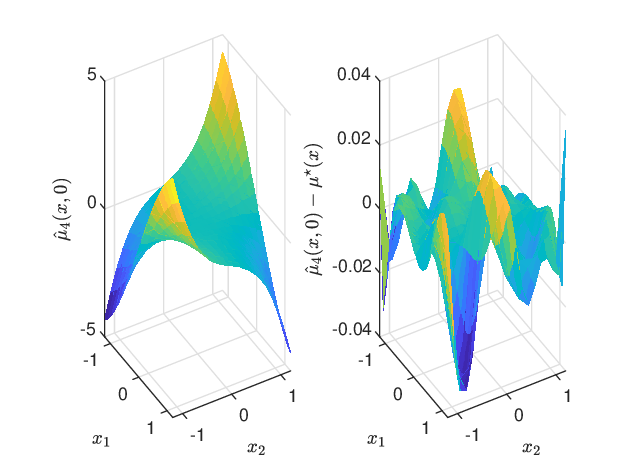}
                \caption{The learnt control policy $\hat{\mu}_{4}(\cdot,0)$ (left) and the error from optimality $\hat{\mu}_{4}(\cdot,0)-\mu^\star(\cdot)$ (right).}
                \label{fig:u10}
\end{figure}

To compute the optimal value function $V^\star$ for this system over the compact set $\Omega=[-1.2,1.2]^2$, we utilize a 3-layer neural network with $250$ nodes in each layer, and with the activation function chosen to be the hyperbolic tangent, and apply the PINN method of Section \ref{sec:PINNs}. We train this network by defining the loss function as the residual \eqref{eq:MSE} of the finite-horizon HJB with horizon length $T=0.5$ and $\phi\equiv0$, with $N_e=10000$ flow data points, and $N_b=N_i=1600$ boundary data points. This loss function is then minimized using Adam \cite{kingma2014adam}, with an initial learning rate equal to $0.01$, decaying at a rate of $0.003$ each time a mini-batch of $500$ flow data points is used. The procedure terminates once the loss function becomes less than $10^{-4}$ for each mini-batch. To account for ill-conditioning at the corners of the state space, the set $\Omega$ during training is chosen to be equal to $[-1.5,1.5]\supset[-1.2,1.2]$.

The results are depicted in Figures \ref{fig:V2}-\ref{fig:u2}. These figures show that while a modest approximation of $V^\star$ and $\mu^\star$ has taken place, the approximation error is quite large. This can also be verified by calculating the value of the steady-state HJB residual \eqref{eq:IHJB_res}, which is equal to $0.7705$; much larger than the desired threshold $10^{-4}$. To improve the accuracy of the approximation, the horizon length is extended to $4$ seconds using the method of Subsection \ref{sub:ext}, in seven increments of $0.5$ seconds. The resulting approximation of the value function $V^\star$ and optimal control $\mu^\star$ is shown in Figures \ref{fig:V10}-\ref{fig:u10},  and is clearly much more accurate than before the horizon length extension. This is also verifiable from the residual \eqref{eq:IHJB_res}, the value of which is $1.536\cdot10^{-4}$.

\subsection{Third-order System}

We consider the following system, borrowed from \cite{third}:
\begin{align*}
\dot{x}_1&=-2x_1+x_2^3,\\
\dot{x}_2&=-Cx_2+x_3,\\
\dot{x}_3&=-K(ax_2^3+bx_2+x_1^3)+Ku,
\end{align*}
with $K=0.05$, $C=3.5$, $a=5.5$, $b=4$.  The aim is to find the optimal control that minimizes \eqref{eq:cost} with $Q(x)=x_1^2+x_2^2+x_3^2$. To do this, we compute the optimal value function $V^\star$ over the compact set $\Omega=[-0.2, 0.2]^3$  by utilizing a 3-layer neural network with $500$ nodes in each layer,  and applying the PINN method with horizon length $T=0.5$, $\phi\equiv0$, $N_e=20000$ flow data points, and $N_b=N_i=27000$ boundary data points. We terminated the training procedure once the loss function became less than $10^{-4}$. 

We look at the value of the infinite-horizon residual \eqref{eq:IHJB_res} as a means to evaluate whether the approximated finite-horizon value function is close enough to the infinite-horizon one. To that end, we calculate the value of this residual for five different choices of $R$, namely $R=100, 10, 1, 0.1, 0.01$, and show the results in Table \ref{tab:2}. We particularly notice an interesting pattern: for small values of $R$, the infinite-horizon residual is very small, in particular, below $10^{-3}$ for $R=0.01$. This is because $R=0.01$ allows the control input to be very authoritative and makes the closed-loop reach the steady state very fast, thus rendering the finite-horizon value function a very good estimate of the infinite-horizon one (see also discussion after Remark 3). As $R$ increases, this property gradually wears off and leads to larger errors \eqref{eq:IHJB_res}.

\begin{table}[!h]\normalsize
  \begin{center}
  \caption{Mean-squared errors for different values of $R$.}\label{tab:2}
   \begin{tabular}{c||c|c|c|c|c}
      \toprule 
      \textbf{Parameter $R$} & \textbf{$0.01$} & \textbf{$0.1$} & \textbf{$1$} & \textbf{$10$} & \textbf{$100$}\\
      \midrule 
      MSE \eqref{eq:IHJB_res} & $0.0008$ & $0.1152$ & $0.0430$ & $0.0031$ & $0.0027$\\
      \bottomrule 
    \end{tabular}
  \end{center}
\end{table}

\section{Conclusion}\label{sec:conc}

We proposed a physics-informed machine learning framework to solve the infinite-horizon optimal control problem of nonlinear systems. Since the steady-state HJB equation has many solutions, we applied PINNs on a finite-horizon HJB that has a unique solution, and which provably uniformly approximates the infinite-horizon value function when the horizon length is large. A method to evaluate whether the horizon length is indeed large was also provided, as well as to extend it (in case it is not) with reduced computations. The robustness of the horizon-extending procedure to approximation errors was also studied, and simulations were performed to verify and clarify the theoretical results.

Future work includes extending this work to solve differential games defined over nonlinear dynamics.

\bibliography{references}

\end{document}